\documentclass[a4paper,12pt]{article}
\usepackage[latin1]{inputenc}
\usepackage[T1]{fontenc}
\usepackage{amsmath, amsthm, amssymb}
\usepackage{array}

\usepackage[final]{graphicx}   
\graphicspath{{tikz/}}

\newtheorem{definition}{Definition}[section]
\newtheorem{theorem}[definition]{Theorem}
\newtheorem{lemma}[definition]{Lemma}
\newtheorem{remark}[definition]{Remark}

\newtheorem{example}[definition]{Example}

\newtheorem{proposition}[definition]{Proposition}
\begin{document}

\pagenumbering{roman}
\pagestyle{empty} 
\title{Complete MDP convolutional codes}
\author{Julia Lieb}
\maketitle



  
\setcounter{page}{1}


\pagenumbering{arabic}
\pagestyle{plain}

\abstract{Maximum distance profile (MDP) convolutional codes have the property that their column distances are as large as possible. It has been shown that, transmitting over an erasure channel, these codes have optimal recovery rate for windows of a certain length. Reverse MDP convolutional codes have the additional advantage that they are suitable for forward and backward decoding algorithms. Beyond that the subclass of complete MDP convolutional codes has the ability to reduce the waiting time during decoding. The first main result of this paper is to show the existence and genericity of complete MDP convolutional codes for all code parameters. The second main contribution is the presentation of two concrete construction techniques to obtain complete MDP convolutional codes. These constructions work for all code parameters but require that the size of the underlying base field is (sufficiently) large.}

\section{Introduction}

Convolutional codes play an important role for digital communication. When considering the erasure channel, which is the most used channel in multimedia traffic, these codes can correct more errors than the classical block codes. 
An erasure channel is a communication channel where the receiver knows if a
received symbol is correct since symbols either arrive correctly or are erased. A prominent example of an erasure channel is the Internet.
The advantage of convolutional codes is the flexibility of grouping the blocks of information in an appropriate way, depending on the erasures location, and then decode the "easy" part of the sequence first, i.e., the part of the sequence with less erasures or where the distribution of erasures allows a
complete correction.

Besides the classical free distance, convolutional codes possess a different notion of distance, called column distance. The column distances of a convolutional code are limited by an upper bound, which was proven in \cite{RS99}. Convolutional codes attaining these bounds, i.e. convolutional codes whose column distances increase as rapidly as possible for as long as possible are called maximum distance profile (MDP) codes. These codes were introduced in \cite{strongly} and are especially suitable for the
use in sequential decoding algorithms. The ability of considering a part of the sequence ("window") of any size and slide this window along the transmitted sequence allows to optimize the number of corrected errors. This feature was firstly discussed in \cite{vp} where the authors showed that MDP convolutional codes have optimal recovery rate for windows of a certain length (depending on the code parameters). Moreover, they considered reverse MDP convolutional codes, which have the advantage that forward and backward decoding algorithms could be used. Finally, complete MDP convolutional codes, which are again a subclass of reverse MDP convolutional codes, have the additional benefit that there is less waiting time when a large burst of erasures occurs and no correction is possible for some time \cite{vp}.

While the existence (and genericity) of reverse MDP convolutional codes for all code parameters has been shown in \cite{vp}, the existence of complete MDP convolutional codes for all code parameters was only conjectured. In this paper, we will prove this conjecture and even show that the property to be complete MDP is a generic property. 

General constructions for all code parameters are only known for MDP convolutional codes \cite{strongly}, \cite{dr13}, but not for reverse and complete MDP convolutional codes.
In this paper, we present two general construction techniques for complete MDP convolutional codes.

The paper is structured as follows. In Section 2, we start with preliminaries about convolutional codes, introduce the notion of column distances as well as MDP and reverse MDP convolutional codes. In Section 3, we examine complete MDP convolutional codes. In the first subsection, we show the existence and genericity of complete MDP convolutional codes for all code parameters and in the second subsection, we present two possibilities to obtain a general construction for complete MDP convolutional codes.

\section{Convolutional Codes}

In this section, we summarize the basic definitions and properties concerning convolutional codes.
One way to define a convolutional code is via polynomial generator matrices.

\begin{definition}\ \\
A \textbf{convolutional code} $\mathfrak{C}$ of \textbf{rate} $k/n$ is a free $\mathbb F[z]$-submodule of $\mathbb F[z]^n$ of rank $k$.
There exists $G\in\mathbb F[z]^{n\times k}$ of full column rank such that
$$\mathfrak{C}=\{v\in\mathbb F[z]^n\ |\ v(z)=G(z)m(z)\ \text{for some}\ m\in\mathbb F[z]^k\}.$$
$G$ is called \textbf{generator matrix} of the code and is unique up to right multiplication with a unimodular matrix $U\in Gl_k(\mathbb F[z])$.\\
The \textbf{degree} $\delta$ of $\mathfrak{C}$ is defined as the maximal degree of the $k\times k$-minors of $G$.
Let $\delta_1,\hdots, \delta_k$ be the column degrees of $G$. Then, $\delta\leq\delta_1+\cdots+\delta_k$ and if $\delta=\delta_1+\cdots+\delta_k$, $G$ is called a \textbf{minimal} generator matrix.
\end{definition}

\begin{remark}\cite[Theorem 2.2]{Fu-He15}\ \\
$G$ is a minimal generator matrix for $\mathfrak{C}$ if and only if it is column proper.
\end{remark}

There is a generic subclass of convolutional codes that could not only be described by an image representation via generator matrices but also by a kernel representation via so-called parity-check matrices, which will be introduced in the following. Therefore, we need the notion of right prime and left prime polynomial matrices.

\begin{definition}\ \\
Let $\overline{\mathbb F}$ denote the algebraic closure of $\mathbb F$.
A polynomial matrix $G\in\mathbb F[z]^{n\times k}$ with $k<n$ is called \textbf{right prime} if it has full column rank for all $z\in\overline{\mathbb F}$. For $k>n$, it is called \textbf{left prime} if it has full row rank for all $z\in\overline{\mathbb F}$. 
\end{definition}

The property to have right prime generator matrices is important for the decoding of a convolutional code.

\begin{definition}\ \\
A convolutional code $\mathfrak{C}$ is called \textbf{non-catastrophic} if one and therefore, each of its generator matrices is right prime.
\end{definition}

 
Non-catastrophic convolutional codes have the property that a finite number of transmission errors could only cause a finite number of decoding errors. Moreover, they have the desired image representation mentioned above.

\begin{definition}\ \\
If $\mathfrak{C}$ is non-catastrophic, there exists a so-called \textbf{parity-check matrix} $H \in\mathbb F[z]^{(n-k)\times n}$ of full rank, such that
$$\mathfrak{C} =\{ v\in \mathbb F[z]^n\ |\ H(z)v(z) = 0 \in \mathbb F[z]^{n-k}\}.$$
Clearly, a parity-check matrix of $\mathfrak{C}$ is not unique and it is possible to choose it left prime and row proper. In this case, the sum of the row degrees of $H$ is equal to the degree $\delta$ of $\mathfrak{C}$ \cite{con}.
\end{definition}

We will need this representation by parity-check matrices to define complete MDP convolutional codes in the following section.
In the remaining part of this section, we want to introduce MDP convolutional codes, for which we have to consider distances of convolutional codes first.

\begin{definition}\ \\
The \textbf{Hamming weight} $wt(v)$ of $v\in\mathbb F^n$ is defined as the number of its nonzero components.\\
For $v\in\mathbb F[z]^n$ with $\deg(v)=\gamma$, write $v(z)=v_0+\cdots+v_{\gamma}z^{\gamma}$ with $v_t\in\mathbb F^n$ for $t=0,\hdots,\gamma$ and set $v_t=0\in\mathbb F^n$ for $t\geq\gamma+1$. Then, for $j\in\mathbb N_0$, the \textbf{j-th column distance} of a convolutional code $\mathfrak{C}$ is defined as
$$d_j^C(\mathfrak{C}):=\min_{v\in\mathfrak{C}}\left\{\sum_{t=0}^j wt(v_t)\ |\ v\not\equiv 0\right\}.$$
Moreover, $d_{free}(\mathfrak{C}):=\lim_{j\rightarrow\infty}d_j^C(\mathfrak{C})$ is called the \textbf{free distance} of $\mathfrak{C}$.
\end{definition}


There exist upper bounds for the free distance and for the column distances of a convolutional code.

\begin{theorem}\cite{RS99}\cite{strongly}\label{ub}
\begin{itemize}
\item[(i)] $d_{free}\leq(n-k)\left(\left\lfloor\frac{\delta}{k}\right\rfloor+1\right)+\delta+1$ (generalized Singleton bound)
\item[(ii)]
 $d_j^C (\mathfrak{C}) \leq (n-k)(j + 1) + 1$
\end{itemize}
\end{theorem}

The bound in part (i) of the preceding theorem is called generalized Singleton bound since for $\delta=0$ one gets the Singleton bound for block codes.\\
We are interested in convolutional codes with good distance properties, i.e. in those codes that reach certain bounds of the preceding theorem.

\begin{definition}\cite{mdp}\ \\
A convolutional code $\mathfrak{C}$ of rate $k/n$ and degree $\delta$ is called\\
 (i) \textbf{maximum distance separable (MDS)} if 
$$d_{free}(\mathfrak{C})=(n-k)\left(\left\lfloor\frac{\delta}{k}\right\rfloor+1\right)+\delta+1,$$
(ii) of \textbf{maximum distance profile (MDP)} if 
$$d_j^C(\mathfrak{C})=(n-k)(j+1)+1\quad \text{for}\ j=0,\hdots,L:=\left\lfloor\frac{\delta}{k}\right\rfloor+\left\lfloor\frac{\delta}{n-k}\right\rfloor$$
(iii) \textbf{strongly maximum distance separable (sMDS)} if 
$$d_M^C(\mathfrak{C})=(n-k)\left(\left\lfloor\frac{\delta}{k}\right\rfloor+1\right)+\delta+1\ \  \text{where}\ \  M=\left\lfloor\frac{\delta}{k}\right\rfloor+\left\lceil\frac{\delta}{n-k}\right\rceil.$$
\end{definition}

As mentioned in the introduction, MDP convolutional codes have the property that their column distances increase as rapidly as possible for as long as possible. Due to the generalized Singleton bound, $j=L$ is indeed the largest possible value for which $d_j^C$ can attain the upper bound from Theorem \ref{ub} (ii). Moreover, the following remark shows that it is sufficient to have equality for $j=L$ in part (ii) of Theorem  \ref{ub} to get an MDP convolutional code.

\begin{remark}\cite{strongly}\ \\
If $d_j^C(\mathfrak{C})=(n-k)(j+1)+1$ for some $j\in\{1,\hdots,L\}$, then $d_i^C(\mathfrak{C})=(n-k)(i+1)+1$ for all $i\leq j$.
\end{remark}

The next remark points out the relationship between MDP, MDS and sMDS convolutional codes.

\begin{remark}\cite{mdp}\ \\
(i) Each sMDS code is an MDS code.\\
(ii) If $n-k$ divides $\delta$, a convolutional code $\mathfrak{C}$ is MDP if and only if it is sMDS.
\end{remark}

In the following, we will provide criteria to check whether a convolutional code is of maximum distance profile. Therefore, we need the notion of trivially zero determinants.\\


\begin{definition}\ \\
Let $n,m\in\mathbb N$ and $A\in\mathbb F^{n\times m}$ be a matrix with the property that each of its entries is either fixed to zero or a free variable from $\mathbb F$. Its determinant $\det(A)$ is called trivially zero if it is zero for all choices for the free variables in $A$.
\end{definition}

\begin{theorem}\cite{strongly}\ \\
Let the convolutional code $\mathfrak{C}$ be generated by a right prime minimal polynomial matrix $G(z)=\sum_{i=0}^{\mu}G_iz^i\in\mathbb F[z]^{n\times k}$ and have the left prime and row proper parity-check matrix $H(z)=\sum_{i=0}^{\nu}H_iz^i\in\mathbb F[z]^{(n-k)\times n}$. Equivalent are:
\begin{itemize}
\item[(i)] $\mathfrak{C}$ is of maximum distance profile.
\item[(ii)] $\mathcal{G}_L:=\left[\begin{array}{ccc} G_0 & & 0\\ \vdots & \ddots &  \\ G_L & \hdots & G_0 \end{array}\right]$ where $G_i\equiv 0$ for $i>\mu$ has the property that every full size minor that is not trivially zero, i.e. zero for all choices of $G_1,\hdots,G_L$, is nonzero.
\item[(iii)] $\mathcal{H}_L:=\left[\begin{array}{ccc} H_0 & & 0\\ \vdots & \ddots &  \\ H_L & \hdots & H_0 \end{array}\right]$ where $H_i\equiv 0$ for $i>\nu$ has the property that every full size minor that is not trivially zero is nonzero.
\end{itemize}
\end{theorem}

\begin{remark}\ \\
The not trivially zero full size minors of $\mathcal{H}_L$ are exactly those which are formed by columns with indices $1\leq j_1<\cdots<j_{(L+1)(n-k)}$ which fulfil $j_{s(n-k)}\leq sn$ for $s=1,\hdots,L$.
\end{remark}

At the end of this section, we introduce reverse MDP convolutional codes, which are advantageous for use in forward and backward decoding algorithms \cite{vp}.

\begin{definition}\cite{h}\ \\
Let $\mathfrak{C}$ be an $(n,k,\delta)$ convolutional code with right prime minimal generator matrix $G$. Set $\overline{g_{ij}(z)}:=z^{\delta_j}g_{ij}(z^{-1})$. Then, the code $\overline{\mathfrak{C}}$ with generator matrix $\overline{G}$ is also an $(n,k,\delta)$ convolutional code, which is called the \textbf{reverse code} to $\mathfrak{C}$.\\
It holds: $v_0+\cdots+v_dz^d\in\overline{\mathfrak{C}}\ \Leftrightarrow\ v_d+\cdots+v_0z^d\in\mathfrak{C}$.
\end{definition}

\begin{definition}\cite{vp}\ \\
Let $\mathfrak{C}$ be an MDP convolutional code. If $\overline{\mathfrak{C}}$ is also MDP, $\mathfrak{C}$ is called \textbf{reverse MDP} convolutional code.
\end{definition}

\begin{remark}\cite{vp}\ \\
Let $(n-k)\mid\delta$ and $H(z) = H_0 + \cdots +H_{\nu}z^{\nu}$ be a left prime and row proper parity-check matrix of the MDP code $\mathfrak{C}$. Then the reverse
code $\overline{\mathfrak{C}}$ has parity-check matrix $\overline{H(z)} = H_{\nu} +\cdots +H_0z^{\nu}$. Therefore, $\mathfrak{C}$ is reverse MDP if and only if every full size minor of the matrix
$$\mathfrak{H}_L:=\left[\begin{array}{ccc} H_{\nu} & \cdots & H_{\nu-L}\\  & \ddots & \vdots \\ 0 &  & H_{\nu} \end{array}\right]$$
formed from the columns with indices $j_1,\hdots,j_{(L+1)(n-k)}$
with $j_{s(n-k)+1} > sn$, for $s = 1,\hdots,L$ is nonzero.
\end{remark}

\section{Complete MDP convolutional codes}

In the beginning of this section, we introduce complete MDP convolutional codes, which are even more advantageous for decoding than reverse MDP convolutional codes \cite{vp}.

\begin{definition}\cite{vp}\label{com}\ \\
Let $H(z)=H_0+H_1z+\cdots H_{\nu}z^{\nu}\in\mathbb F[z]^{(n-k)\times n}$ be a parity-check matrix of the convolutional code $\mathfrak{C}$ of rate $k/n$. Set $L:=\lfloor\frac{\delta}{n-k}\rfloor+\lfloor\frac{\delta}{k}\rfloor$. Then
\begin{align}\label{ppc}
\mathfrak{H}:=\left(\begin{array}{ccccc}
H_{\nu} & \cdots & H_0 &   & 0 \\ 
  & \ddots &   & \ddots &   \\ 
0 &   & H_{\nu} & \cdots & H_0
\end{array}\right)  \in\mathbb F^{(L+1)(n-k)\times (\nu+L+1)n}
\end{align}
is called \textbf{partial parity-check matrix} of the code. Moreover, $\mathfrak{C}$ is called \textbf{complete MDP} convolutional code if for any of its parity-check matrices $H$, every full size minor of $\mathfrak{H}$ which is not trivially zero is nonzero.
\end{definition}

\begin{remark}\cite{vp}\ \\
Every complete MDP convolutional code is a reverse MDP convolutional code.
\end{remark}

As for $\mathcal{H}_L$ - when considering MDP convolutional codes - and additionally for $\mathfrak{H}_L$ - when considering reverse MDP convolutional codes - one could describe the not trivially zero full size minors of the partial parity-check matrix $\mathfrak{H}$ by conditions on the indices of the columns one uses to form the corresponding minor.

\begin{lemma}\cite{vp}\label{index}\ \\
A full size minor of $\mathfrak{H}$ formed by the columns $j_1,\hdots,j_{(L+1)(n-k)}$ is not trivially zero if and only if 
\begin{itemize}
\item[(i)]
$j_{(n-k)s+1}>sn$
\item[(ii)]
$j_{(n-k)s}\leq sn+\nu n$
\end{itemize}
for $s=1,\hdots,L$.
\end{lemma}

The following lemma enables us to show the existence and genericity of complete MDP convolutional codes in Section 3.1 as well as to provide concrete constructions in Section 3.2 by considering only the not trivially full size minors of a matrix $\mathfrak{H}$ of the form \eqref{ppc}.

\begin{lemma}\label{htoc}\ \\
Let $H(z)=H_0+H_1z+\cdots H_{\nu}z^{\nu}\in\mathbb F[z]^{(n-k)\times n}$ be such that each full size minor of $\mathfrak{H}$ as in \eqref{ppc} which is not trivially zero is nonzero. Then $H$ is a row proper parity-check matrix of an $(n,k,\delta)$ complete MDP convolutional code, where $\delta=\nu(n-k)$. In particular, for an $(n,k,\delta)$ complete MDP convolutional code, it always holds $(n-k)\mid\delta$.
\end{lemma}

\begin{proof}\ \\
If one sets $s=L$ in part (ii) of Lemma \ref{index}, one sees that there are not trivially zero full size minors of $\mathfrak{H}$ that are formed by a set of columns which contains $n-k$ of the last $n$ columns. Therefore, $H_0$ is of full row rank, which implies that $H\in\mathbb F[z]^{(n-k)\times n}$ is of full row rank. Hence $H$ is the parity-check matrix of a convolutional code with rate $k/n$.\\
If one sets $s=1$ in part (i) of Lemma \ref{index}, one obtains that there are not trivially zero full size minors of $\mathfrak{H}$ that are formed by a set of columns which contains $n-k$ of the first $n$ columns. Thus, $H_{\nu}$ has full row rank. In particular, it contains no row that consists only of zeros and hence, all $n-k$ row degrees of $H$ are equal to $\nu$. Consequently, $\delta=\nu(n-k)$.
\end{proof}

\subsection{Existence and genericity of complete MDP convolutional codes}

The existence of MDP and reverse MDP convolutional codes for all code parameters has been proven in \cite{mdp} and \cite{vp}, respectively, by showing that the sets of these codes are Zariski open in the quasi-projective variety of all non-catastrophic $(n,k,\delta)$ convolutional codes. Moreover, this implies that the sets of MDP and reverse MDP convolutional codes form generic subsets of this variety. In the following, we show that this is also true for complete MDP convolutional codes.

\begin{theorem}\ \\
Let $n,k,\delta\in\mathbb N$ with $k<n$ and $(n-k)\mid\delta$. 
Then, the set of all $(n,k,\delta)$ complete MDP convolutional codes forms a generic subset of the variety of all non-catastrophic $(n,k,\delta)$ convolutional codes.
In particular, there exists an $(n,k,\delta)$  complete MDP convolutional code over a sufficiently large base field. 
\end{theorem}

\begin{proof}\ \\
The set of non-catastrophic $(n,k,\delta)$ convolutional codes with parity-check matrix $H$ whose row degrees are all equal to $\nu:=\frac{\delta}{n-k}$ is Zariski open and therefore dense in the set of all non-catastrophic $(n,k,\delta)$ convolutional codes; see e.g. \cite{s01}. Hence, we could assume that $H$ has this so-called "generic" row degrees.\\
Consider the set of all polynomial matrices $H\in\overline{\mathbb F}[z]^{(n-k)\times k}$ with $\deg(H)\leq\nu$. For each choice of columns from $\mathfrak{H}$ such that the index conditions from Lemma \ref{index} are fulfiled, there is $H\in\overline{\mathbb F}^{(n-k)\times k}[z]$ with $\deg(H)\leq\nu$ such that the corresponding minor of $\mathfrak{H}$ is nonzero. Moreover, similarly to the proof for Theorem 2.7 of \cite{mdp}, one could argue that the set of such matrices $H$ for which this minor is nonzero is Zariski open in the set of all $H\in\overline{\mathbb F}^{(n-k)\times k}[z]$ with $\deg(H)\leq\nu$ since the entries of the coefficient matrices of $H$ fulfil a polynomial equation in $\overline{\mathbb F}[x_1,\hdots,x_{(\nu+1)(n-k)n}]$ if the minor is zero.\\
Forming the intersection of all non-empty and Zariski open sets that correspond to a minor whose columns fulfil the index conditions of Lemma \ref{index}, one gets that the set of all $H\in\overline{\mathbb F}^{(n-k)\times k}[z]$ with $\deg(H)\leq\nu$ for which all not trivially zero minors are nonzero is non-empty and Zariski open in the set of all $H\in\overline{\mathbb F}^{(n-k)\times k}[z]$ with $\deg(H)\leq\nu$. All matrices within this non-empty and open set are parity-check matrices of an $(n,k,\delta)$ complete MDP convolutional code with all row degrees equal to $\nu$ (see Lemma \ref{htoc}) and this set is also non-empty and open in the set of all $H\in\overline{\mathbb F}^{(n-k)\times k}[z]$ which are parity-check matrices of an $(n,k,\delta)$ convolutional code with all row degrees equal to $\nu$.
Therefore, the set of all $(n,k,\delta)$ complete MDP convolutional codes forms a generic subset of all non-catastrophic $(n,k,\delta)$ convolutional codes. Furthermore, this implies that there exists a sufficiently large field $\mathbb F$ for which an $(n,k,\delta)$ complete MDP  convolutional code exists.
\end{proof}

\begin{remark}\ \\
Since the set of all non-catastrophic $(n,k,\delta)$ convolutional codes is Zariski open in the variety of all $(n,k,\delta)$ convolutional codes, the set of all $(n,k,\delta)$ complete MDP convolutional codes is also a generic subset of the variety of all $(n,k,\delta)$ convolutional codes.
\end{remark}

\subsection{Construction of complete MDP convolutional codes}

The proof for the existence (and genericity) of complete MDP convolutional codes for all code parameters in the preceding subsection was non-constructive. In this subsection, we will present two concrete construction techniques for complete MDP convolutional codes. These work for all code parameters but require that the size of the underlying field is sufficiently large.\\
For the first construction, we apply the following lemma, which considers matrices over $\mathbb Z$, and use that these matrices could also be viewed as matrices over $\mathbb F_p$ if the characteristic $p$ is sufficiently large.

\begin{lemma}\cite{strongly}\label{X}\ \\
For $a,b\in\mathbb N$ with $b<a$, let $X:=\left[\begin{array}{cccc}
1 & &  & 0 \\ 
1 & \ddots &  \\ 
 & \ddots & \ddots \\ 
0 &  & 1 & 1
\end{array}\right]\in\mathbb Z^{a\times a}$ and $\hat{X}:=(X^b)^{i_1,\hdots,i_r}_{j_1,\hdots,j_r}$ be constructed out of rows $1\leq i_1<\cdots<i_r\leq a$ and columns $1\leq j_1<\cdots<j_r\leq a$ of $X^b=\left[\begin{array}{ccccccc}
1 &   &   &   &   & 0 \\ 
\binom{b}{1} & \ddots &  &   &  &  \\ 
\vdots & \ddots & \ddots  &  &   &   \\ 
\binom{b}{b-1} &  & \ddots & \ddots &   &   \\ 
1 & \ddots &  & \ddots   & \ddots  &   \\ 
& \ddots & \ddots & & \ddots & \ddots \\
0  & & 1 & \binom{b}{b-1} & \cdots & \binom{b}{1} & 1
\end{array}\right]$. Then, $\det(\hat{X})\geq 0$ and $\det(\hat{X})>0\Leftrightarrow j_l\in\{i_l-b,\hdots,i_l\}$ for $l=1,\hdots,r$.
\end{lemma}

In the following, we give a general construction for $(n,k,\delta)$ complete MDP convolutional codes based on the preceding lemma. Doing this, we can assume $\nu=\frac{\delta}{n-k}$; see Lemma \ref{htoc}.

\begin{theorem}\ \\
With the notation from the preceding lemma, choose $X\in\mathbb F^{(\nu+L+1)n\times(\nu+L+1)n}$, i.e. $a:=(\nu+L+1)n$, as well as $b:=\nu n+k$. For $j=0,\hdots,L$, set $I_j=\{(\nu+j)n+k+1,\hdots, (\nu+j+1)n\}$ and $I=\bigcup_{j=0}^L I_j$.\\
Then, those rows of $X^b$ whose indices lie in $I$ form the partial parity-check matrix of an $(n,k,\delta)$  complete MDP convolutional code if the characteristic of the base field is greater than $\binom{\nu n+k}{\lfloor 1/2(\nu n+k)\rfloor}^{(n-k)(L+1)}\cdot((n-k)(L+1))^{1/2(n-k)(L+1)}$.
\end{theorem}

\begin{proof}\ \\
Defining the partial parity-check matrix as in the theorem, one gets\\
\scriptsize
 $\mathfrak{H}=$
$$\left[\begin{array}{cccccccccccccc}
1 & \nu n+k & \hdots & \binom{\nu n+k}{n-1} & \hdots & \binom{\nu n+k}{k} & \hdots & 1 & 0 &  &  &  &  & 0 \\ 
&  &  & \vdots &  & \vdots &  &  &  &  &  & &  &  \\ 
0 & 1 & \hdots & \binom{\nu n+k}{k} & \hdots & \binom{\nu n+k}{n-1} & \hdots & \nu n+k & 1 &  &  &  &  &  \\ 
 &  &  & & \ddots &  &  &  &  & \ddots &  &  &  &  \\ 
&  &  &  &  & 1 & \nu n+k & \hdots & \binom{\nu n+k}{n-1} & \hdots & \binom{\nu n+k}{k} & \hdots & 1 & 0 \\ 
 &  &  &  &  &  &  &  & \vdots &   & \vdots &  &  & \\ 
0 &  &  &  &  & 0 & 1 & \hdots & \binom{\nu n+k}{k} & \hdots & \binom{\nu n+k}{n-1} & \hdots & \nu n+k & 1
\end{array} \right]$$

\normalsize

i.e. $H_{\nu}=\left[\begin{array}{ccccc}
1 &  & \hdots &  & \binom{\nu n+k}{n-1} \\ 
 & \ddots &  &  & \vdots \\ 
0 &  & 1 & \hdots & \binom{\nu n+k}{k}
\end{array}\right],\hdots, H_0=\left[\begin{array}{ccccc}
\binom{\nu n+k}{k} & \hdots & 1 &  & 0 \\ 
\vdots &   &  & \ddots &  \\ 
\binom{\nu n+k}{n-1}  &  & \hdots  &  & 1
\end{array}\right]$.

Write $I=\{i_1,\hdots,i_{(n-k)(L+1)}\}$ with $i_1<\cdots<i_{(n-k)(L+1)}$.
Using Lemma \ref{index} and Lemma \ref{X}, it only remains to show that the conditions $j_{(n-k)s+1}>sn$ and $j_{(n-k)s}\leq sn+\nu n$ for $s=1,\hdots,L$ are equivalent to $j_l\in\{i_l-(\nu n+k),\hdots,i_l\}$ for $l=1,\hdots,(n-k)(L+1)$. But this is true since both are equivalent to $j_l\in\{l+kx,\hdots,l+kx+\nu n+k\}$
where $x\in\{0,\hdots, L\}$ is chosen such that $l\in\{x(n-k)+1,\hdots,(x+1)(n-k)\}$.

The necessary field characteristic size is estimated similar to \cite{strongly}. The determinant of an $A\times A$ matrix with largest entry equal to $B$ is upper bounded by $B^A\cdot A^{A/2}$. Setting $B=\binom{\nu n+k}{\lfloor\frac{\nu n+k}{2}\rfloor}$ and $A=(n-k)(L+1)$ yields the stated result.
\end{proof}

\begin{remark}\ \\
\vspace{-6mm}
\begin{itemize}
\item[(i)]
The construction of the preceding theorem simply means to skip the first $\nu n+k$ rows of $X^b$ and then, alternately choose $n-k$ rows and skip $k$ rows of $X^b$. 
\item[(ii)] The bound for the size of the characteristic in the preceding theorem is not very sharp. In fact, much smaller sizes are possible; see e.g.\ the following example.
\end{itemize}
\end{remark}

In the following, we illustrate the construction technique from the preceding theorem with the help of two examples.

\begin{example}\ \\
\vspace{-6mm}
\begin{itemize}
\item[(a)] Consider the case $(n,k,\delta)=(3,2,1)$, i.e. $\nu=1$, $L=1$, $\nu n+k=5$ and $(\nu+L+1)n=9$. The parity-check matrix $H(z)=H_0+H_1z$ with $H_0=[10\ 5\ 1]$ and $H_1=[1\ 5\ 10]$ defines a $(3,2,1)$ complete MDP convolutional code over each finite field $\mathbb F_{p^n}$ with characteristic $p\not\in\{2,3,5,11\}$. If one uses the bound from the preceding theorem, one gets that $H$ defines a complete MDP convolutional code over fields with characteristic greater than $10^2\cdot 2^1=200$, which shows that this bound is not sharp.
\item[(b)]
Consider the case $(n,k,\delta)=(3,1,4)$, i.e. $\nu=2$, $L=6$, $\nu n+k=7$ and $(\nu+L+1)n=21$.
A $(3,1,4)$ complete MDP convolutional code is given through the parity-check matrix $H(z)=H_0+H_1z+H_2z^2$ with $H_0=\left[\begin{array}{ccc}
7 & 1 & 0 \\ 
21 & 7 & 1
\end{array}\right]$, $H_1=\left[\begin{array}{ccc}
35 & 35 & 21 \\ 
21 & 35 & 35
\end{array}\right]$ and $H_2=\left[\begin{array}{ccc}
1 & 7 & 21 \\ 
0 & 1 & 7
\end{array}\right]$ over a field with characteristic greater than $35^{14}\cdot 14^7 \approx 4,36\cdot 10^{29}$. One could see that the bound for the characteristic rapidly increases with the sizes of the code parameters.
\end{itemize}
\end{example}

In the second part of this subsection, we want to present a second construction technique for complete MDP convolutional codes. It also requires large field sizes but has the advantage that it works for arbitrary characteristic of the underlying field. For this construction, we need the following definition and proposition.

\begin{definition}\cite{dr16}\ \\
%
%
Let $S_n$ be the symmetric group of order $n$. The determinant of an $n\times n$ matrix $A=[a_{i,l}]$ is given by $\det(A)=\sum_{\sigma\in S_n}(-1)^{sgn(\sigma)}a_{1,\sigma(1)}\cdots a_{n,\sigma(n)}$.
We call a product of the form $a_{1,\sigma(1)}\cdots a_{n,\sigma(n)}$ with $\sigma\in S_n$ a \textbf{trivial term} of the determinant if at least one component $a_{i,\sigma(i)}$ is equal to zero.

\end{definition}


\begin{proposition}\cite[Theorem 3.3]{dr16}\ \\
Let $\alpha$ be a primitive element of a finite field $\mathbb F=\mathbb F_{p^N}$ and $B=[b_{i,l}]$ be a matrix over $\mathbb F$ with the following properties
\begin{enumerate}
\item if $b_{i,l}\neq 0$, then $b_{i,l}=\alpha^{\beta_{i,l}}$ for a positive integer $\beta_{i,l}$
\item if $b_{i,l}=0$, then $b_{i',l}=0$ for any $i'>i$ or $b_{i,l'}=0$ for any $l'<l$
\item if $l<l'$, $b_{i,l}\neq 0$ and $b_{i,l'}\neq 0$, then $2\beta_{i,l}\leq\beta_{i,l'}$
\item if $i<i'$, $b_{i,l}\neq 0$ and $b_{i',l}\neq 0$, then $2\beta_{i,l}\leq\beta_{i',l}$.
\end{enumerate}
Suppose $N$ is greater than any exponent of $\alpha$ appearing as a nontrivial term of any minor of $B$. Then $B$ has the property that each of its minors which is not trivially zero is nonzero.
\end{proposition}

\begin{remark}\ \\
The preceding proposition even implies that each minor that is not trivially zero is nonzero, not only those of full size, which we need for our construction of complete MDP convolutional codes.
\end{remark}

The next theorem provides a general construction of complete MDP convolutional codes based on the preceding proposition.

\begin{theorem}\ \\
Let $n,k,\delta\in\mathbb N$ with $k<n$ and $(n-k)\mid\delta$ and let $\alpha$ be a primitive element of a finite field $\mathbb F=\mathbb F_{p^N}$ with $N>(L+1)\cdot 2^{(\nu+2)n-k-1}$. Then $H(z)=\sum_{i=0}^{\nu}H_iz^i$ with $H_i=\left[\begin{array}{ccc}
\alpha^{2^{in}} & \hdots & \alpha^{2^{(i+1)n-1}} \\ 
\vdots &  & \vdots \\ 
\alpha^{2^{(i+1)n-k-1}} & \hdots & \alpha^{2^{(i+2)n-k-2}}
\end{array}\right]$ for $i=0,\hdots,\nu=\frac{\delta}{n-k}$ is the parity-check matrix of an $(n,k,\delta)$ complete MDP convolutional code.
\end{theorem}

\begin{proof}\ \\
We have to show that each fullsize minor of the partial parity-check matrix $\mathfrak{H}$ given by \eqref{ppc} that is not trivially zero is nonzero. Permutation (reverse ordering) of the blocks of columns of $\mathfrak{H}$, which does not change the terms for the not trivially zero fullsize minors, leads to the matrix $\left[\begin{array}{ccccc}
0 &  & H_0 & \hdots & H_{\nu} \\ 
 & \text{\reflectbox{$\ddots$}} &  & \text{\reflectbox{$\ddots$}} &  \\ 
H_0 & \hdots & H_{\nu} &  & 0
\end{array}\right]$, which fulfills the four conditions of the preceding proposition. Consequently, $H$ is the parity-check matrix of a complete MDP convolutional code if $N$ is greater than any exponent of $\alpha$ appearing as a nontrivial term of any minor of $\mathfrak{H}$. The largest possible value for such an exponent is $(L+1)\cdot(2^{(\nu+2)n-k-2}+2^{(\nu+2)n-k-4}+\cdots+2^{\nu n+k})=(L+1)\cdot 2^{\nu n+k}\sum_{i=0}^{n-k-1}2^{2i}<(L+1)\cdot 2^{(\nu+2)n-k-1}$. 
\end{proof}

We conclude this section by considering examples for this second construction principle.

\begin{example}\ \\
\vspace{-6mm}
\begin{itemize}
\item[(a)]
As in the first example for the first construction technique, we construct a $(3,2,1)$ complete MDP convolutional code, i.e. $\nu=L=1$. One gets $H_0=\left[\begin{array}{ccc}
\alpha & \alpha^2 & \alpha^4 
\end{array}\right]$ and $H_1=\left[\begin{array}{ccc}
\alpha^8 & \alpha^{16} & \alpha^{32} 
\end{array} \right]$ over $\mathbb F_{p^N}$ with $N>2^7$. This requires a field size greater than $2^{128}$, which is much larger than for the first construction technique but the characteristic of the underlying field can be chosen arbitrarily.
\item[(b)]
For the construction of a $(3,1,4)$ complete MDP convolutional code, we obtain the parity-check matrix $H(z)=H_0+H_1z+H_2z^2$ with $H_0=\left[\begin{array}{ccc}
\alpha & \alpha^2 & \alpha^4 \\ 
\alpha^2 & \alpha^4 & \alpha^8
\end{array}\right]$, $H_1=\left[\begin{array}{ccc}
\alpha^8 & \alpha^{16} & \alpha^{32} \\ 
\alpha^{16} & \alpha^{32} & \alpha^{64}
\end{array} \right]$ and $H_2=\left[\begin{array}{ccc}
\alpha^{64} & \alpha^{128} & \alpha^{256} \\ 
\alpha^{128} & \alpha^{256} & \alpha^{512}
\end{array} \right]$ over $\mathbb F_{p^N}$ with $N>7\cdot2^{10}$. This leads to a field size of at least $2^{7\cdot 2^{10}}=2^{7168}$, which is again even much larger than for the other construction technique.
\end{itemize}
\end{example}

\section{Conclusion}

In this paper, the existence and genericity of complete MDP convolutional codes for all code parameters has been shown and two general construction techniques have been provided. However, these construction techniques require underlying base fields of very large sizes.
This provokes the question if it is possible to derive general constructions over fields of smaller size and what is the minimum required field size so that such a construction is possible. This problem is even unsolved for MDP not only for complete MDP convolutional codes. In \cite{b} a bound for the existence of superregular matrices was proved and used to obtain an upper bound on the necessary field size for MDP convolutional codes. But small examples show that these bounds are not very sharp. Furthermore, no constructions achieving these bounds were provided. Hence, these remain open problems for future research.

\bibliography{mybibfile}

\begin{thebibliography}{}




\bibitem{dr13}
Almeida, PJ; Napp, D; Pinto, R (2013) A new class of superregular matrices and MDP convolutional codes. Linear Algebra and its Applications, 439:2145-2157


\bibitem{dr16}
Almeida, PJ; Napp, D; Pinto, R (2016) Superregular matrices and applications to convolutional codes. Linear Algebra and its Applications 499:1-25



\bibitem{Fu-He15} Fuhrmann PA, Helmke U (2015) The Mathematics of Networks of Linear Systems. Springer, New York





\bibitem{strongly}
Gluesing-Luerssen, H; Rosenthal, J; Smarandache, R (2006) Strongly-MDS Convolutional Codes. IEEE Transactions on Information Theory 52.2:584-598


\bibitem{h} 
Hutchinson, R (2008) The existence of strongly MDS convolutional codes. SIAM J. Control Optim. 47.6:2812-2826


\bibitem{mdp} Hutchinson, R; Rosenthal, J; Smarandache, R (2005) Convolutional codes with maximum distance profile. Systems $\&$ Control Letters 54:53-63

\bibitem{b}  Hutchinson, R; Smarandache, R; Trumpf, J (2008) On superregular matrices and MDP convolutional codes. Linear Algebra and its Applications 428:2585-2596





\bibitem{con}
Rosenthal, J (2001) Connections between linear systems and convolutional
codes. In B. Marcus and J. Rosenthal, editors, Codes, Systems and
Graphical Models. IMA Vol. 123. Springer, p. 39-66



\bibitem{RS99}
Rosenthal, J; Smarandache, R (1999) Maximum distance separable convolutional codes. Appl. Algebra Engrg. Comm. Comput. 10.1:15-32


\bibitem{s01}
Smarandache, R; Gluesing-Luerssen, H; Rosenthal, J (2001) Constructions of MDS-convolutional codes. IEEE Transactions on Information Theory 47.5:2045-2049

\bibitem{vp} Tomas, V; Rosenthal, J; Smarandache, R (2012) Decoding of Convolutional Codes Over the Erasure Channel. IEEE Transactions on Information Theory 58.1:90-108 




\end{thebibliography}

\end{document}